\documentclass[11pt, a4paper]{article}
\usepackage[english]{babel}
\usepackage[latin1]{inputenc}
\usepackage[T1]{fontenc}
\usepackage{amsmath}
\usepackage{amsfonts}
\usepackage{amssymb}
\usepackage{amsthm}
\usepackage{appendix}
\usepackage{enumitem}
\usepackage{dsfont}
\usepackage{times}
\usepackage{bbm}
\usepackage[affil-it]{authblk}
\usepackage[left=3.5cm, right=3.5cm, bottom=4.3cm]{geometry}
\newtheorem{theorem}{Theorem}[section]
\newtheorem{lemma}[theorem]{Lemma}
\newtheorem{proposition}[theorem]{Proposition}

\theoremstyle{definition}

\theoremstyle{remark}
\newtheorem*{remark}{Remark}

\numberwithin{equation}{section}
\newcommand{\tr}{\operatorname{tr}}
\newcommand{\supp}{\operatorname{supp}}
\newcommand{\ran}{\operatorname{ran}}
\newcommand{\rank}{\operatorname{rank}}
\newcommand{\spec}{\operatorname{spec}}
\newcommand{\specd}{\operatorname{spec}_\mathrm{drop}}
\newcommand{\spa}{\operatorname{span}}

\newcommand{\chiln}{\mathcal{X}_\Lambda^n}

\newcommand{\hd}{\mathcal{H}_\Lambda^I}
\newcommand{\x}{\textbf{x}}
\newcommand{\y}{\textbf{y}}
\newcommand{\z}{\textbf{z}}
\newcommand{\w}{\textbf{w}}
\newcommand{\0}{\emptyset}
\newcommand{\B}{\textbf{B}}

\usepackage{color}

\begin{document}
\title{Bounds on the entanglement entropy of droplet states in the XXZ spin chain}
\author{V. Beaud\footnote{vincent.beaud@ma.tum.de}, S. Warzel}
\affil{Zentrum Mathematik, Boltzmannstr.~3, Technische Universit\"at M\"unchen}
\date{September 29, 2017}
\maketitle

\begin{abstract}
	We consider a class of one-dimensional quantum spin systems on the finite lattice $\Lambda\subset\mathbbm{Z}$, related to the XXZ spin chain in its Ising phase. It includes in particular the so-called droplet Hamiltonian. The entanglement entropy of energetically low-lying states over a bipartition $\Lambda = B \cup B^c$ is investigated and proven to satisfy a logarithmic bound in terms of $\min\lbrace n,\vert B\vert,\vert B^c\vert\rbrace$, where $n$ denotes the maximal number of down spins in the considered state. Upon addition of any (positive) random potential the bound becomes uniformly constant on average, thereby establishing an area law. The proof is based on spectral methods: a deterministic bound on the local (many-body integrated) density of states is derived from an energetically motivated Combes--Thomas estimate.
	\vspace{0.7cm}
\end{abstract}

\maketitle

\section{Introduction}

The bipartite entanglement entropy (EE) is one prominent measure for the correlation structure of a many-body state~\cite{CalCar09,Eisert:2010uq}. In the context of quantum spin systems defined over the vertex set $ \Lambda $ of a finite, but arbitrarily large graph, pure many-body states $ \varrho = | \psi \rangle \langle \psi | $ are identified with normalized vectors $ \psi $ from the tensor product Hilbert space $ \mathcal{H}_\Lambda = \bigotimes_{x \in \Lambda} \mathbbm{C}^\nu $, where $ \nu \in \mathbbm{N} $ denotes the Hilbert space dimension attached to a single vertex. Carving out a subset $ B \subset \Lambda $, the state's information related to observables with support in $B$ is contained in the reduced state $ \varrho_B = \tr_{\mathcal{H}_{B^c}} \varrho $, where the trace is over the tensor component associated with the complement of $ B $ in $ \Lambda $. The family of R\'enyi entropies 
\begin{equation}\label{def:renyi}
	S_\alpha(\psi, B) := \frac{1}{1-\alpha} \ln\tr\left[ \left( \varrho_B\right)^\alpha\right],
\end{equation}
with parameter $ \alpha \in (0,\infty)$, then serves as a non-negative measure of the entanglement of the pure state $ \varrho = |\psi \rangle \langle \psi|$ over the bipartition $ B \cup B^c $. Since $S_\alpha(\psi,B)=S_\alpha(\psi,B^c)$ and to simplify notation, all bounds involving $B$ shall in the sequel implicitly be understood as the minimum over $B$ and $B^c$. The expression~\eqref{def:renyi} coincides with the von Neumann entropy $ S_1(\psi, B) = - \tr \varrho_B \ln \varrho_B $ in the limit $ \alpha \to 1 $ and with the Hartley entropy $ S_0(\psi, B) = \ln \rank \varrho_B $ in the limit $ \alpha \to 0 $. The latter is also referred to as max-entropy since it upper-bounds all other entropies. More generally, it follows from Jensen's inequality that the R\'enyi entropies are monotonically decreasing in $ \alpha $. 

Clearly, any product state has zero entropy over any bipartition. At the other extreme, the EE of any pure state is trivially  bounded by $ S_0(\psi, B)  \leqslant \ln \dim \mathcal{H}_B  = | B | \, \ln \nu $, which grows linearly in the volume $ |B | $. Such a volume dependence is in fact generic for most states as is demonstrated by Page's law on the average entropy of subsystems~\cite{Page:1993ys,Foong:1994kx}. In constrast, the EE of states with a short-range correlation structure is expected to scale with the surface area $ |\partial B | $. Such an area law bound is evident for matrix-product states and their higher dimensional generalizations, the projected entanglement pairs. For the ground state of any gapped spin chain it was derived in the celebrated work of Hastings~\cite{Hastings:2007vn}. Recently, it has been established that exponential clustering of the state is indeed sufficient for the validity of an area law bound in spin chains~\cite{Brandao:2015fk}. Yet another behavior is found for the EE of the ground state of free fermions for which $ S_\alpha(\psi, B) $ scales asymptotically as $  | \partial B | \ln | B | $ for any $ \alpha \in (0,\infty) $ \cite{Wolf:2006fp,Gioev:2006uq,Helling:2011kx}. Finally, quantum field theoretic methods successfully predict the universal scaling of the EE at one-dimensional critical points and relate it to the central charge of the underlying conformal field theory~\cite{Calabrese09}.

\subsection{The XXZ chain and its Ising phase}

In this paper, we will be concerned with energetically low-lying states of XXZ spin chains. The energy of this system without external fields is given by the Hamiltonian
\begin{equation}
	h_\Lambda =   - \Delta^{-1} \sum_{x=-L}^{L-1} \left( S_x^{1} S_{x+1}^{1} + S_x^{2} S_{x+1}^{2}  \right)-\sum_{x=-L}^{L-1} \left(  S_x^{3} S_{x+1}^{3} - \tfrac{1}{4}  \right) + \beta(\Delta)\left( 1 -  S_{-L}^{3} - S_{L}^{3} \right)  \label{eq:HXXZ}
\end{equation}
defined on the Hilbert space $ \mathcal{H}_{\Lambda} = \bigotimes_{x \in \Lambda} \mathbbm{C}^2 $ with $ \Lambda = [-L,L] \cap \mathbbm{Z} $, $L\in\mathbbm{N}$, and $\beta(\Delta) =\frac{1}{2}(1-\Delta^{-1})$. The Pauli spin matrices $ S^{j} $, $ j \in \{ 1,2,3\} $, are normalized to have eigenvalues $ \pm \tfrac{1}{2} $ with corresponding normalized eigenvectors denoted by $ \{ | \frac{1}{2} \rangle , | -\frac{1}{2}  \rangle \} \subset \mathbbm{C}^2 $. The subscript $ x $ indicates that such a matrix is lifted to $  \mathcal{H}_{\Lambda} $ and acts non-trivially only on the $ x $-component of the tensor product. The parameter $ \Delta^{-1} \in [0,\infty) $ measures the anisotropy of the interaction, with $\Delta^{-1} = 1$ covering the isotropic Heisenberg ferromagnet and $\Delta^{-1} = 0$ the Ising case. We will be concerned throughout with the Ising phase of the model, i.e., the regime $ \Delta^{-1} \in [0,1) $. The chosen value of the boundary field, $\beta(\Delta) = \frac{1}{2}(1-\Delta^{-1})$, ensures that down spins at the boundary of $\Lambda$ are not energetically favored. Moreover, it is the minimal value such that the threshold inequalities \eqref{eq:thresholds} underlying the Combes--Thomas estimate \eqref{eq:CT} remain valid.

Two particularities in the spectrum of $h_\Lambda$ deserve to be emphasized. On the one hand, $ h_\Lambda $ commutes with the total magnetization $ S_{\rm tot}^{3} = \sum_{x=-L}^L S_x^3 $, whence its spectrum can be classified according to the eigenvalues of the latter. A natural orthonormal eigenbasis of $S_{\rm tot}^{3}$ then consists of product states $ | \mathbf{\sigma} \rangle = \bigotimes_{x\in \Lambda} | \sigma_x\rangle $ with $ \sigma_x \in \{ \pm  \frac{1}{2}  \} $. The eigenspace $ \mathcal{H}_\Lambda^n $ of  $  S_{\rm tot}^{3} $  with eigenvalue $ \frac{|\Lambda |}{2} - n $ is spanned by states $\vert\mathbf{\sigma}\rangle$  with a constant number $ n \in \{0, 1, 2 , \dots , |\Lambda | \}  $ of down spins. The Hamiltonian $ h_\Lambda $ is then block diagonal with respect to the decomposition $\mathcal{H}_{\Lambda} = \bigoplus_{n=0}^{|\Lambda|} \mathcal{H}_\Lambda^n $ and its spectrum thus consists of the union
	\begin{equation*}
		\spec(h_\Lambda) = \bigcup_{n=0}^{|\Lambda|}\spec\bigl(h_\Lambda\vert_{\mathcal{H}_\Lambda^n}\bigr).
	\end{equation*}
On the other hand, $\spec(h_\Lambda)$ exhibits a series of thresholds. In fact, under the dynamics generated by $h_\Lambda$, spin configurations with down spins arbitrarily distributed over $\Lambda$ are energetically less favorable than configurations with all down spins next to each other, if $\Delta^{-1}\in[0,1)$. In other words, breaking up clusters of down spins costs energy. Accordingly, the Hamiltonian restricted to the span of product states $\vert\sigma\rangle$ with at least $k$ down-spin clusters, $k=1,\ldots,\vert\Lambda\vert$, admits increasing lower bounds (cf.~Lemma~\ref{lemma:thresholds}). More precisely, let us define
	\begin{equation}\label{def:C}
		\mathcal{C}:=\bigl\lbrace\sigma\in\lbrace\pm\tfrac{1}{2}\rbrace^{\vert\Lambda\vert}\; \big\vert\;
			\textnormal{all down spins in }\sigma\textnormal{ form a chain of nearest neighbors}\bigr\rbrace,
	\end{equation}
	the spin configurations with fully clustered down spins, and denote by $\mathcal{Q}$ the orthogonal projection onto $\mathcal{H}_\Lambda\setminus\ell^2(\mathcal{C})$, the subspace generated by vectors with non-clustered down spins. Most relevant for our analysis are the following results.
\begin{proposition}\label{Prop:generalities}
	In the Ising phase $ \Delta^{-1} \in [0, 1) $, we have:
	\begin{enumerate}[label=(\roman*)]
		\item $ h_\Lambda \geqslant 0 $. Zero is a simple eigenvalue of $ h_\Lambda $ with eigenvector $\bigotimes_{x\in \Lambda}|\tfrac{1}{2}\rangle$.
		\item $\mathcal{Q}h_\Lambda\mathcal{Q}\geqslant 2(1-\Delta^{-1})\mathcal{Q}$.
	\end{enumerate}
\end{proposition}
\begin{remark}
	 While the existence of zero as a simple eigenvalue can be verified by direct computation, both inequalities follow as special cases of Lemma~\ref{lemma:thresholds}. In essence, they rely on the observation that $h_\Lambda$ may be rewritten as a sum of two positive terms, the negative lattice Laplacian and a potential term, which is kept as a lower bound (see for instance~\cite{BW17,Stolz}).
\end{remark}
The energy range of interest to our investigations is the \textit{droplet spectrum} or \textit{droplet band} of $h_\Lambda$ defined as
\begin{equation}\label{def:droplet_band}
	\specd(h_\Lambda):=\spec(h_\Lambda)\cap I,\qquad I=\bigl[0,2(1-\Delta^{-1})\bigr).
\end{equation}
Normalized $\psi\in\ran P_I(h_\Lambda)$ are referred to as \textit{droplet states}, where $P_I(h_\Lambda)$ denotes the spectral projection of $h_\Lambda$ associated to $I$. For reasons inherent to technicalities in the proof of the Combes--Thomas estimate \eqref{eq:CT}, the results presented in this article are restricted to
\begin{equation}\label{def:intervalI}
	I := [0,E_1]\qquad\textnormal{for any fixed }\quad E_1 < 2(1-3\Delta^{-1}).
\end{equation}
This effectively reduces the range of the anisotropy parameter to $\Delta^{-1}\in[0,1/3)$. As illustrated below in the instance of the Ising limit $\Delta^{-1}=0$, the droplet subspace is essentially generated by clustered states $\vert\sigma\rangle$, with $\sigma\in\mathcal{C}$. This aligns in particular with the fact that the energy $2(1-\Delta^{-1})$ corresponds to the two-cluster breakup. Furthermore, as a consequence of the Combes--Thomas estimate~\eqref{eq:CT}, droplets states are clustered, up to exponentially small corrections. More precisely,
\begin{equation*}
	\vert\psi(\mathbf{\sigma})\vert \leqslant C e^{-\mu\,d(\mathbf{\sigma},\mathcal{C})},\qquad\psi\in\ran P_I(h_\Lambda),
\end{equation*}
for some $C,\mu\in(0,\infty)$ independent of $\Lambda$ and $d(\mathbf{\sigma},\mathcal{C})$ the $\ell^1$-distance from $\mathbf{\sigma}$ to the subset of configurations with clustered down spins $\mathcal{C}$ (see Lemmas~\ref{lemma:CT} and~\ref{lemma:density_states}).

It is instructive to consider two special cases: the Ising limit $\Delta^{-1} = 0$ and the classical droplet Hamiltonian~\cite{NS01,NSS07,FS14}, where the boundary field in~\eqref{eq:HXXZ} is set to $\beta(\Delta) = \frac{1}{2}\sqrt{1-\Delta^{-2}}$.
\medskip\\
\noindent \textbf{Ising limit.} The product states $|\sigma \rangle$ are also eigenvectors of $ h_\Lambda $ with eigenvalues given by $ E(\sigma) = -   \sum_{x=-L}^{L-1} \left(  \sigma_x \sigma_{x+1} - \frac{1}{4}  \right) +  \frac{1}{2}\left( 1 -  \sigma_{-L} - \sigma_{L} \right) $. The subspace corresponding to the droplet regime, i.e., linear combinations of spin configurations $\sigma$ with energies $E(\sigma) < 2$, then coincides with the direct sum of subspaces
\begin{equation*}
	\mathcal{D}^n_\Lambda = \spa\left\{ | \mathbf{\sigma} \rangle \in   \mathcal{H}_\Lambda^n\;
		\bigl\vert\;\sigma \in \mathcal{C}  \right\},\qquad n\in\lbrace 0,1,2,\ldots,\vert\Lambda\vert\rbrace,
\end{equation*}
with the convention $ \mathcal{D}^0_\Lambda =  \mathcal{H}_\Lambda^0 $. The EE of any linear combination of such droplets is trivially bounded in terms of the dimension of the image of the reduced state on $ B \subset \Lambda $. The range is however also restricted to the span of:
\begin{itemize}
	\item droplet states in $ \mathcal{H}_B $ which either have $ n = 0 $ or $ n = | B | $ down spins or have a droplet of at most 
$ |B|-1 $ down spins at a boundary of $B$.
	\item the canonical embedding of the projection of $ \psi $ on droplet states in  $ \mathcal{H}_B $.
\end{itemize}
The dimension of this space is at most $ 2 + \vert\partial B\vert(|B|-1) + 1$. This yields:
\begin{proposition}\label{Prop:Ising}
	For all normalized vectors $ \psi \in \bigoplus_{n=0}^{|\Lambda|} \mathcal{D}^n_\Lambda $ and all intervals $ B \subset \Lambda \subset\mathbbm{Z}$, we have
	\begin{equation}\label{eq:boundIsing}
		S_0(\psi;B) \leqslant  \ln \bigl(3 + \vert\partial B\vert(|B|-1)\bigr),
	\end{equation}
	where $\partial B := \bigl\lbrace (x,y)\;\big\vert\; \vert x-y\vert=1, x\in B, y\in B^c\bigr\rbrace$ denotes the boundary of $B$. 
\end{proposition}
The bound \eqref{eq:boundIsing} may be further tightened when restricting to normalized $\psi^{(n)}$ with at most $n$ down spins, i.e., $\psi^{(n)}\in\bigoplus_{k=0}^n\mathcal{D}^k_\Lambda$. The same arguments as stated above then lead to
\begin{equation}\label{eq:boundIsingn}
	S_0(\psi^{(n)};B) \leqslant  \ln \bigl(3 + 2(\min\lbrace n,|B|\rbrace-1)\bigr). 
\end{equation}
This bound is asymptotically optimal and saturated if $\psi^{(n)}$ is for instance chosen as a superposition of clusters $\vert\sigma\rangle$, $\sigma\in\mathcal{C}$, of $n$ down spins uniformly distributed over $\Lambda$. In fact, also a superposition of clusters uniformly distributed over $B$ and its outer boundary will do.
\medskip\\
\noindent\textbf{Classical droplet Hamiltonian.} This Hamiltonian, $H_\Lambda^d$, equals the right-hand side of~\eqref{eq:HXXZ} with $\beta(\Delta)=\frac{1}{2}\sqrt{1-\Delta^{-2}}$. The chosen value of the boundary field naturally arises when considering spectral properties of $H_\Lambda^d$ in the thermodynamic limit $\Lambda\to\mathbbm{Z}$. Here, it appears as an instance of the Hamiltonian $H_\Lambda$ below, where positive potential terms are added to $h_\Lambda$. Its spectrum has been studied in detail on the subspaces $\mathcal{H}_\Lambda^n$ with exactly $n\in\lbrace 0,1,2,\ldots,\vert\Lambda\vert\rbrace$ down spins~\cite{NS01,NSS07,FS14}:
\begin{enumerate}
	\item Zero is a simple eigenvalue of $H_\Lambda^d\vert_{\mathcal{H}_\Lambda^0}$.
	\item For sufficiently large $\Delta$ or sufficiently large $n$, the $\vert \Lambda\vert - n +1 $ lowest eigenvalues of $H_\Lambda^d\vert_{\mathcal{H}_\Lambda^n}$ are concentrated in a neighborhood of $\sqrt{1-\Delta^{-2}}$, whose width decreases to $0$ as $n\to\infty$. These eigenvalues are separated from the rest of the spectrum by a gap, which increases in $n$ and reaches the limit $1-\Delta^{-1}$ as $n\to\infty$.
\end{enumerate}
The union over $n=0,1,\ldots,\vert\Lambda\vert$ of these eigenvalues are traditionally referred to as droplet spectrum or droplet band of $H_\Lambda^d$. For sufficiently large $\Delta$, this in fact coincides with our previous definition \eqref{def:droplet_band}, where the droplet band consisted of all eigenvalues between $0$ and $2(1-\Delta^{-1})$. The dimension $\vert\Lambda\vert -n +1$ of the droplet subspace associated to $H_\Lambda^d\vert_{\mathcal{H}_\Lambda^n}$ corroborates the previous argument that droplet states are essentially generated by clustered states $\vert\sigma\rangle$, $\sigma\in\mathcal{C}$.

\subsection{Main result}

In view of the results on the stability of the area law for the EE in gapped spin chains~\cite{Marien:2016fk}, it is not surprising to learn that a bound of the form~\eqref{eq:boundIsingn} remains valid (at least for $ \alpha = 1 $) for the droplet Hamiltonian $H_\Lambda^d$, see Theorem~\ref{thm:entropy} below. After all, the above characterization of the droplet spectrum guarantees that a gap is open for sufficiently large $\Delta$. Remarkably, such a bound remains true for a much larger class of Hamiltonians related to the XXZ spin chain in its Ising phase. More precisely, our result pertains to adding an arbitrary non-negative term to $h_\Lambda$:
\begin{equation}
	H_\Lambda := h_\Lambda + V_\Lambda \qquad\textnormal{with}\quad V_\Lambda = \sum_{x\in\Lambda} b_x \left(\frac{1}{2} - S^3_x\right) \, , \qquad b_x \geqslant 0. 
\end{equation}
The droplet spectrum is now defined as $\specd(H_\Lambda)= \spec(H_\Lambda)\cap I$ with $I = [0,E_1]$ for any fixed $E_1<2(1-3\Delta^{-1})$ and the subspace of droplet states is $\hd = \ran P_I(H_\Lambda)$. The potential term $V_\Lambda$ may---and generally will---pollute a possible gap present above $\specd(H_\Lambda)$. However, the EE of arbitrary droplet states is still bounded as suggested by Proposition~\ref{Prop:Ising}.
\begin{theorem}\label{thm:entropy}
	Let $I=[0,E_1]$ for some $E_1<2(1-3\Delta^{-1})$. Then, for any $\alpha>0$, there exist $a_\alpha, b_\alpha,C_\alpha \in (0,\infty)$ such that for all non-negative fields $\lbrace b_x\rbrace_{x\in\Lambda}$ and all normalized $\psi^{(n)}\in\hd$ with at most $0\leqslant n\leqslant |\Lambda|$ down spins
	\begin{equation}\label{eq:entropy_general}
		S_\alpha(\psi^{(n)},B) \leqslant C_\alpha\ln\bigl(a_\alpha + b_\alpha\min\lbrace n,\vert B\vert\rbrace\bigr)
	\end{equation}
	for all finite intervals $B\subset\Lambda\subset\mathbbm{Z}$.
\end{theorem}
\noindent It is worthwhile noting that:
\begin{enumerate}
	\item The inequality (\ref{eq:entropy_general}) is merely an upper bound and there are states in $\hd$ with much lower entanglement entropy. As a product state, the ground state $\bigotimes_{x\in\Lambda}\vert\frac{1}{2}\rangle$  has for instance $S_\alpha = 0$. 
	\item Even though the above statement concerns the low-energy spectrum of $H_\Lambda$, it is not restricted to the lowest (non-degenerate) eigenvalue. The droplet states include the ground state of $H_\Lambda$ but are, in general, \emph{not} ground states.
		\item Depending on $\lbrace b_x\rbrace_{x\in\Lambda}$, the subspace of droplet states $\hd$ may be empty and the above assertion thus trivial.
	\item As already suggested by the Ising case $\Delta^{-1}=0$ in Proposition~\ref{Prop:Ising}, there is no area law for general droplet states $\psi\in\hd$. Moreover, the logarithmic bound is sharp and saturated by a superposition of droplets (with $n$ down spins) uniformly distributed over $B$. The intuition is again by counting states and noting that droplets have effective dimension one.

The scaling behavior of the EE of ground states in other spin chains such as the XY chain or, more generally, the XXZ chain in a critical regime $ |\Delta^{-1}| \geqslant 1 $ of the anisotropy parameter  is known to be logarithmic in $ | B | $ \cite{Latorre09}. This should however not be confused with the logarithmic bound~\eqref{eq:entropy_general}, which does not result from the criticality of eigenstates but rather from the effective dimension of the droplet space.
	\item Quantum quench protocols study the time-evolution $\psi_t = e^{-itH_\Lambda}\psi$ of initial vectors
	$ \psi $, which are usually prepared as  low-entanglement eigenstates of some Hamiltonian other then the evolution's generator. The bound~\eqref{eq:entropy_general} does not apply to $ \psi_t $ for arbitrary  initial states, but for the large class of low-energy states $\psi\in\hd$, whose one-dimensional subspaces are generally not invariant under $H_\Lambda $. Theorem~\ref{thm:entropy} then shows that the  dynamical EE  of $ \psi_t $ is bounded uniformly in $t\in\mathbbm{R}$.
\end{enumerate}
The result of Theorem~\ref{thm:entropy} can be significantly improved if the collection of (non-negative) potentials is made random. 
\begin{theorem}\label{thm:entropy_random}
	Let the setting be as in Theorem~\ref{thm:entropy}. Then, for any $\alpha>0$ and any $0<\epsilon<\min\lbrace\alpha,1\rbrace$, there exists $C_\epsilon\in(0,\infty)$ such that for all collections $\lbrace b_x\rbrace_{x\in\Lambda}$ of iid random variables with non-trivial support $\supp b_x\subset [0,\infty)$
	\begin{equation}\label{eq:entropy_random}
		\mathbb{E}\left[\sup\limits_{\substack{\psi\in\hd,\Vert\psi\Vert =1}}
			\exp\bigl\lbrace (1-\epsilon)S_\alpha(\psi,B)\bigr\rbrace\right]\leqslant C_\epsilon
	\end{equation}
	for all finite intervals $B\subset\Lambda\subset\mathbbm{Z}$.
\end{theorem}
\noindent We continue with a few comments.
\begin{enumerate}
		\item A similar system was studied in~\cite{BW17,Stolz}, where exponential localization of states in $\hd$ is established at strong disorder. It should be emphasized that the above area law does \emph{not} require the localization result as a premise. It rather appears as a consequence of the deterministic result in Theorem~\ref{thm:entropy} and of exponential suppression of large droplet states. In particular, the result already holds for any non-trivial random potential and does not need strong disorder.
	\item On the one hand,  the averaged EE  of random quantum critical systems is expected to grow logarithmically in $ |B | $ \cite{Refael09,Eisert:2010uq}. On the other hand, a strong disordered external field may result in an area law for the EE of eigenstates in the many-body localized (MBL) phase in accordance with the presence of local integrals of motion. In fact, such an area law was proven 	for eigenstates of disordered free fermions as well as of the disordered quantum XY spin chain~\cite{Pastur14,Elgart2017,Abdul-Rahman:2016hl,ASSN}. 
		\item Theorem~\ref{thm:entropy_random} applies in particular to the entanglement dynamics, where the supremum is taken over $t\in\mathbbm{R}$ and $\psi_t = e^{-itH_\Lambda}\psi$ with $\psi\in\hd$. The uniform boundedness of the entanglement dynamics is reminiscent of recent results concerning the quantum XY chain in disorder~\cite{ASSN}. A putative similar behavior was hinted at for another exactly solvable case, the disordered quantum harmonic oscillators, where the full time-evolution of initial product states were shown to have averaged correlations decaying exponentially in space, uniformly in time~\cite{Sims17}.
Note that a fundamentally different behavior is expected in the localized phase of the  antiferromagnetic XXZ spin chain, where the EE was argued both numerically and theoretically to grow  logarithmically in time~\cite{Prosen08,Pollmann12,Serbyn13,Huse14}.

\end{enumerate}

\section{The model and its properties}

\subsection{Hard-core particle formulation}

Based on the block-diagonal form of $H_\Lambda$ with respect to the decomposition $\mathcal{H}_{\Lambda} = \bigoplus_{n=0}^{|\Lambda|} \mathcal{H}_\Lambda^n $ of the Hilbert space into sectors with constant number of down spins, the system admits an interpretation in terms of hard-core particles. Let
\begin{equation*}
	\chiln :=  \left\lbrace\, \x =\lbrace x_1,x_2,\ldots,x_n\rbrace\in\Lambda^n \; \big| \;
		x_1<x_2<\ldots <x_n\,\right\rbrace 
\end{equation*}
denote the configuration space of $n$ hard-core particles on $\Lambda$ and define $\mathcal{X}_\Lambda$ as the disjoint union
\begin{equation*}
	\mathcal{X}_\Lambda := \bigcup\limits_{n=0}^{|\Lambda|} \chiln.
\end{equation*}
We adopt the convention that $\mathcal{X}_\Lambda^0 = \lbrace\emptyset\rbrace$, where $\emptyset$ denotes the empty configuration. The identification of each down spin with a hard-core particle induces a natural unitary equivalence between $\mathcal{H}^n_\Lambda$ and $\ell^2(\chiln)$, and thus between $\mathcal{H}_\Lambda$ and $\ell^2(\mathcal{X}_\Lambda)$. The natural inner product on $\ell^2(\mathcal{X}_\Lambda)$ will be denoted by $\langle \cdot,\cdot\rangle$ and an orthonormal basis is given by $\delta_\x$, $\x\in\mathcal{X}_\Lambda$, with $\delta_\x(\y)=\delta_{\x,\y}$.

Equivalently, configuration $\x\in\chiln$ may be understood as a subset of $n$ sites in $\Lambda$. Elements of $\x$ forming an isolated string of neighboring sites in $\Lambda$ will be called a cluster. We denote by $\mathcal{C}_\Lambda^{(k)}$ the subset of configurations in $\mathcal{X}_\Lambda$ with exactly $1\leqslant k\leqslant |\Lambda|$ clusters. For later convenience, the shorthand notation $\mathcal{C}\equiv\mathcal{C}^{(1)}_\Lambda$ is used. By the unitary equivalence, the latter corresponds to the spin characterization of clusters~\eqref{def:C}. For notational ease, we shall henceforth stick to the hard-core formulation.

\subsection{Energy thresholds and Combes-Thomas estimate}

The dynamics generated by $H_\Lambda$ favors clustering, as reflected in two results from~\cite{BW17}. The first provides an increasing sequence of lower bounds for the restriction of the Hamiltonian to subspaces with at least $k$ clusters. The second employs this fact to establish a Combes--Thomas estimate on states (essentially) below the two-cluster breakup.
\begin{lemma}[cf.~Lemma 1.1~in~\cite{BW17}]\label{lemma:thresholds}
	Let $\mathcal{Q}^{(k)}$ be the orthogonal projection onto the subspace $\bigoplus_{j=k}^{|\Lambda|}\ell^2\bigl(\mathcal{C}_\Lambda^{(j)}\bigr)$. Then, for all non-negative fields $\lbrace b_x\rbrace_{x\in\Lambda}$, we have
	\begin{equation}\label{eq:thresholds}
		\mathcal{Q}^{(k)} H_\Lambda \mathcal{Q}^{(k)} 
			\geqslant \mathcal{Q}^{(k)} h_\Lambda \mathcal{Q}^{(k)}
			\geqslant k (1-\Delta^{-1})\mathcal{Q}^{(k)}.
	\end{equation}
\end{lemma}

\begin{lemma}[cf.~Theorem 2.1.~in~\cite{BW17}]\label{lemma:CT}
	Let $I = [0,E_1]$ for some $E_1<2(1-3\Delta^{-1})$. Then, there exist $\mu_\textsc{t},C_\textsc{t}\in(0,\infty)$ such that for all $E\in I$, all $0\leqslant n\leqslant |\Lambda|$ and all $\x,\y\in\chiln \setminus\mathcal{C}$
	\begin{equation}\label{eq:CT}
		\left\vert G^{(2)}_\Lambda(\x,\y;E)\right\vert 
			:= \left\vert\left\langle\delta_\x, \bigl(\mathcal{Q}^{(2)}(H_\Lambda -E)\mathcal{Q}^{(2)}\bigr)^{-1}\delta_\y\right\rangle\right\vert
			\leqslant C_\textsc{t} e^{-\mu_\textsc{t}\,d(\x,\y)},
	\end{equation}
	where $d(\x,\y) := \sum_{k=1}^n\vert x_k-y_k\vert$ denotes the $\ell^1$-distance between configurations $\x$ and $\y$.
\end{lemma}
\begin{remark}
	The requirement $\x,\y\in\chiln \setminus\mathcal{C}$ could be replaced by $\x,\y\in\mathcal{X}_\Lambda \setminus\mathcal{C}$, as the left-hand side of \eqref{eq:CT} vanishes for $\x$ and $\y$ in subspaces $\chiln$ with different $n$.
\end{remark}

\section{Entanglement bounds on droplet states}

\subsection{Deterministic bound on the local density of states}

Both Theorems~\ref{thm:entropy} and \ref{thm:entropy_random} rely on a deterministic bound on the local (many-body integrated) density of states
\begin{equation}
	N_I(\x) := \langle\delta_\x, P_I(H_\Lambda)\delta_\x\rangle.
\end{equation}
\begin{lemma}\label{lemma:density_states}
	Let $I=[0,E_1]$ with $E_1 < 2(1-3\Delta^{-1})$. Then, there exist $\mu,C\in(0,\infty)$ such that, for all normalized $\psi\in \hd$, all finite intervals $\Lambda\subset\mathbbm{Z}$ and all $\x\in\mathcal{X}_\Lambda$
	\begin{equation}\label{eq:density_states}
		\vert \psi(\x)\vert^2\leqslant N_I(\x) \leqslant C e^{-\mu\,d(\x,\mathcal{C})},
	\end{equation}
	where $d(\x,\mathcal{C}) := \min\lbrace d(\x,\y)\; \vert\; \y\in\mathcal{C} \rbrace$ denotes the $\ell^1$-distance of the configuration $\x$ to the fully clustered configurations $\mathcal{C}$.
\end{lemma}
\begin{proof}
	As $\psi\in \hd$, the first inequality in~\eqref{eq:density_states} follows by an immediate application of the Cauchy-Schwarz inequality. For the second inequality, notice first that the statement is trivial for clustered $\x\in\mathcal{C}$, since  $N_I(\x)\leqslant 1$; whence henceforth $\x\not\in\mathcal{C}$. Furthermore, as $H_\Lambda$ has discrete spectrum, it is always possible to choose a larger interval $I=[0,E_1]$ satisfying $E_1<2(1-3\Delta^{-1})$ and with $E_1$ in a spectral gap, $E_1\not\in\spec(H_\Lambda)$. We rewrite the projection operator $P_I(H_\Lambda)$ as a contour integral
	\begin{equation*}
		P_I(H_\Lambda) = \frac{1}{2\pi i}\int_\Gamma\frac{1}{z-H_\Lambda}\,\mathrm{d}z,
	\end{equation*}
	where $\Gamma$ is the positively oriented closed path along the rectangle in the complex plane with corners $E_1\pm i\alpha$ and $-1\pm i\alpha$ for some $\alpha>0$. We treat the two vertical and horizontal edges separately. Let $\mathcal{P}$ be the projection onto the clustered configurations and $\mathcal{Q}\equiv\mathcal{Q}^{(2)}=\mathds{1}-\mathcal{P}$. We split
	\begin{equation*}
		H_\Lambda = \bigl[\mathcal{P}H_\Lambda\mathcal{P} + \mathcal{Q}H_\Lambda\mathcal{Q}\bigr]
			+ \bigl[\mathcal{Q}H_\Lambda\mathcal{P} + \mathcal{P}H_\Lambda\mathcal{Q}\bigr] =: H_0 + H_1
	\end{equation*}
	and use the resolvent equation
	\begin{equation*}
		\mathcal{Q}\frac{1}{z-H_\Lambda}\mathcal{Q}
			= r_\mathcal{Q}(z) + r_\mathcal{Q}(z)\, \mathcal{Q}H_\Lambda\mathcal{P}\,
				\frac{1}{z-H_\Lambda}\,\mathcal{P}H_\Lambda\mathcal{Q}\, r_\mathcal{Q}(z),
	\end{equation*}
	where $r_\mathcal{Q}(z) = \bigl(\mathcal{Q}(z-H_\Lambda)\mathcal{Q}\bigr)^{-1}$. By Lemma~\ref{lemma:thresholds}, $r_\mathcal{Q}(z)$ is analytic in the interior of $\Gamma$, whence its contour integral along $\Gamma$ vanishes. It follows that for any $\x\not\in\mathcal{C}$
	\begin{equation*}
		N_I(\x) = \frac{1}{2\pi i}\int_\Gamma\langle\delta_\x,\frac{1}{z-H_\Lambda}\delta_\x\rangle\,\mathrm{d}z
			= \frac{1}{2\pi i}\int_\Gamma\langle\psi_{(\x,z)},\frac{1}{z-H_\Lambda}\psi_{(\x,z)}\rangle\,\mathrm{d}z,
	\end{equation*}
	where $\psi_{(\x,z)} = \mathcal{P}H_\Lambda\mathcal{Q}\,r_\mathcal{Q}(z)\delta_\x$. Notice that $\psi_{(\x,z)}$ is supported on clustered configurations $\y\in\mathcal{C}\cap\chiln$ with exactly $n$ down spins for all $\x\in\chiln$. By the Combes--Thomas estimate \eqref{eq:CT}, we have
	\begin{equation*}\label{bound_psi}
		\vert\psi_{(\x,z)}(\y)\vert \leqslant \sum\limits_{\delta_\w\in\mathrm{ran}\mathcal{Q}}
				\vert H_\Lambda(\y,\w)\vert \vert G_\Lambda^{(2)}(\w,\x;z)\vert
			\leqslant 2C_\textsc{t}e^{\mu_\textsc{t}}e^{-\mu_\textsc{t}\,d(\x,\y)}.
	\end{equation*}
	Nonvanishing contributions in the sum have $d(\y,\w)=1$, $H_\Lambda(\y,\w)=1$. For each given $\y\in\mathcal{C}$, at most two $\w$ satisfy those conditions. Note that the bound is independent of $z\in\Gamma$. Setting $C:=2C_\textsc{t}e^{\mu_\textsc{t}}$, we infer
	\begin{align}
		\Vert\psi_{(\x,z)}\Vert^2 &=\, \sum\limits_{\y\in\mathcal{C}\cap\chiln}\vert\psi_{(\x,z)}(\y)\vert^2
			\leqslant C^2\,e^{-\mu_\textsc{t}\,d(\x,\mathcal{C})}
				\sum\limits_{\y\in\mathcal{C}\cap\chiln}e^{-\mu_\textsc{t}\vert x_1 -y_1\vert} \notag \\
			&\leqslant\, C^2\,\coth\left(\frac{\mu_\textsc{t}}{2}\right)e^{-\mu_\textsc{t}\,d(\x,\mathcal{C})}.
			\label{eq:bound_psi2}
	\end{align}
	The first inequality is by $2d(\x,\y)\geqslant\vert x_1-y_1\vert+d(\x,\mathcal{C})$ and the second by an explicit computation. Any $z$ on the two horizontal edges of the contour $\Gamma$ has constant imaginary part $\pm\alpha$ and thus there exists a finite $C$ such that
	\begin{equation}\label{eq:bound_psi3}
		\vert \langle\psi_{(\x,z)},\frac{1}{z-H_\Lambda}\psi_{(\x,z)}\rangle\vert
			\leqslant \Vert\psi_{(\x,z)}\Vert^2 \Vert\frac{1}{z-H_\Lambda}\Vert
			\leqslant \frac{C}{\alpha}\, e^{-\mu_\textsc{t}\,d(\x,\mathcal{C})}.
	\end{equation}	
	The two vertical edges of the contour $\Gamma$ are parametrized by $-1 + i\eta$ and $E_1 + i\eta$, $\eta\in[-\alpha,\alpha]$. We only consider the line $E_1+i\eta$ and denote it by $\Gamma_1$, as the treatment of $-1+i\eta$ is similar to \eqref{eq:bound_psi3}. Using the resolvent identity	$r_\mathcal{Q}(E_1+i\eta) = r_\mathcal{Q}(E_1) - i\eta\,r_\mathcal{Q}(E_1)\,r_\mathcal{Q}(E_1+i\eta)$	and setting $\varphi_{(\x,E_1+i\eta)} := \mathcal{P}H_\Lambda\mathcal{Q}\,r_\mathcal{Q}(E_1)\,r_\mathcal{Q}(E_1+i\eta)\delta_\x$, we have
	\begin{align}
		\langle\psi_{(\x,z)},\frac{1}{z-H_\Lambda}\psi_{(\x,z)}\rangle
			=\,&\, \langle\psi_{(\x,E_1)},\frac{1}{z-H_\Lambda}\psi_{(\x,E_1)}\rangle\label{eq_expansion_1} \\
			&\,-\, \langle\psi_{(\x,z)},\frac{i\eta}{z-H_\Lambda}\varphi_{(\x,z)}\rangle
				+\langle\varphi_{(\x,z)},\frac{i\eta}{z-H_\Lambda}\psi_{(\x,E_1)}\rangle.\label{eq_expansion_2}
	\end{align}
	By \eqref{eq:bound_psi2}, the integral along $\Gamma_1$ of the right-hand side of \eqref{eq_expansion_1} is bounded by
	\begin{equation*}
		\Vert \int_{-\alpha}^\alpha \frac{\mathrm{d}\eta}{E_1-H_\Lambda+i\eta}\Vert\,\Vert\psi_{(\x,E_1)}\Vert^2
			\leqslant C e^{-\mu_\textsc{t}\,d(\x,\mathcal{C})},
	\end{equation*}
	where $C$ is some finite constant independent of $\alpha>0$. This follows by applying the spectral theorem to $\bigl(E_1-H_\Lambda+i\eta\bigr)^{-1}$ and noting that for any $\epsilon\neq 0$ we have
	\begin{equation*}
		\vert\int_{-\alpha}^{\alpha}\frac{\mathrm{d}\eta}{\epsilon +i\eta}\vert \leqslant \vert \operatorname{arg}(\epsilon+i\alpha) - \operatorname{arg}(\epsilon-i\alpha)\vert\leqslant 2\pi.
	\end{equation*}
	
	We now turn to the two terms \eqref{eq_expansion_2}. As $\Vert\mathcal{P}H_\Lambda\mathcal{Q} \Vert \leqslant 2$ and $\Vert r_\mathcal{Q}(E_1+i\eta)\Vert$ is uniformly bounded in $\eta$, so is $\Vert\varphi_{(\x,E_1+i\eta)}\Vert$. Moreover, since $\Vert i\eta / (z-H_\Lambda)\Vert \leqslant 1$, there exists a finite constant $C$ such that the integral along $\Gamma_1$ of each of the two terms in \eqref{eq_expansion_2} is bounded by
	\begin{equation*}
		2\alpha \Vert\varphi_{(\x,z)}\Vert\,\Vert\psi_{(\x,z)}\Vert
			\leqslant 2\alpha\, C e^{-\mu\,d(\x,\mathcal{C})}
	\end{equation*}
	with $\mu=\mu_\textsc{t}/2$. This concludes the proof.
\end{proof}

\subsection{Deterministic entropy bound}

The proof of Theorem~\ref{thm:entropy} combines the result of Lemma~\ref{lemma:density_states} with some summability properties of factors of the form $\exp\bigl[-\mu\, d(\x,\mathcal{C})\bigr]$.
\begin{proof}[Proof of Theorem~\ref{thm:entropy}]
	Let $0\leqslant n\leqslant |\Lambda|$. We consider a normalized vector $\psi\in\bigoplus_{k=0}^n\mathcal{H}_\Lambda^k$ with $\rho =\vert\psi\rangle\langle\psi\vert$ and $\rho_B=\tr_{\mathcal{H}_{B^c}}\rho$ on the (non-trivial) interval $B\subset\Lambda$. By the monotonicity of the R\'enyi entropy, we can restrict to $\alpha\in(0,1)$. Moreover, to establish \eqref{eq:entropy_general}, it suffices to show
	\begin{equation*}
		\tr\bigl[(\rho_B)^\alpha\bigr] \leqslant c_\alpha s^\gamma
	\end{equation*}
	for some ($\alpha$-dependent) constants $c_\alpha,\gamma\in(0,\infty)$. The reduced state $\rho_B$ has matrix elements
	\begin{equation*}
		\rho_B(\x,\y) = \sum\limits_{\z\in\mathcal{X}_{B^c}} \overline{\psi(\lbrace\x,\z\rbrace)}\psi(\lbrace\y,\z\rbrace)
	\end{equation*}
	for configurations $\x,\y\in\mathcal{X}_B$ with at most $s:=\min\lbrace n,\vert B\vert\rbrace$ down spins. By the normalization $\tr\rho_B =1$, every diagonal element satisfies $\rho_B(\x,\x)\leqslant 1$. In a first step, we distinguish the cases $s=\vert B\vert$ and $s=n<\vert B\vert$. If $s=\vert B\vert$, Jensen's inequality for $\alpha\in(0,1)$ yields
	\begin{equation}\label{eq:bound_sB}
		\tr\bigl[(\rho_B)^\alpha\bigr] \leqslant\sum\limits_{\x\in\mathcal{X}_B^{\leqslant\vert B\vert}}\rho_B(\x,\x)^\alpha
			\leqslant 2 + \sum\limits_{\x\in\mathcal{X}_B^{<\vert B\vert}\setminus\lbrace\0\rbrace}\rho_B(\x,\x)^\alpha.
	\end{equation}
	For the last inequality, we excluded from the sum the diagonal elements with $\x=\0$, corresponding to all spins up on $B$, and $\x=\B$, corresponding to all spins down on $B$, and estimated them trivially.
	
	Let now $s=n<\vert B\vert$ and denote by  $P_n$ the orthogonal projection onto the subspace $\mathcal{H}_B^n$ with exactly $n$ down spins. Observe that the non-trivial matrix elements of $\rho_B P_n$ satisfy
	\begin{equation*}
		(\rho_B P_n)(\x,\y) = \overline{\psi(\lbrace\x,\0\rbrace)}\psi(\lbrace\y,\0\rbrace)
			= \overline{\psi_B(\x)}\psi_B(\y),\qquad \x\in\mathcal{X}_B,\y\in\mathcal{X}_B^n,
	\end{equation*}
	where $\psi_B= \sum_{\x\in\mathcal{X}_B}\psi(\x)\delta_\x$ denotes the restriction of $\psi$ to its components on $B$. Hence, $\rho_B P_n$ has rank one and may be cast as $\rho_B P_n = \overline{\rho}_B P_n$ with $\overline{\rho}_B = \vert\psi_B\rangle\langle\psi_B\vert$. Similarly $P_n\rho_B = P_n\overline{\rho}_B$. Setting $Q_n = 1-P_n$, we have the decomposition
	\begin{align*}
		\rho_B &=\, Q_n \rho_B Q_n + Q_n \overline{\rho}_B P_n + P_n \overline{\rho}_B Q_n + P_n \overline{\rho}_B P_n
			=: Q_n \rho_B Q_n + \widehat{\rho}_B.
	\end{align*}
	Since $\overline{\rho}_B$ has rank one, $\vert\widehat{\rho}_B\vert$ has rank at most two and thus $\tr\bigl[\vert\widehat{\rho}_B\vert^\alpha\bigr]\leqslant 2$ for any $\alpha\in (0,1)$. We conclude that if $n<\vert B\vert$
	\begin{align*}
		\tr\bigl[(\rho_B)^\alpha\bigr]	
			&\leqslant\,\left(\tr\bigl[(Q_n\rho_B Q_n)^\alpha\bigr]
				+ \tr\bigl[\vert\widehat{\rho}_B\vert^\alpha\bigr]\right)
			\leqslant 2 + \tr\bigl[(Q_n\rho_B Q_n)^\alpha\bigr]
	\end{align*}
	and by Jensen's inequality
	\begin{equation}\label{eq:bound_sn}
		\tr\bigl[(Q_n\rho_B Q_n)^\alpha\bigr]
			\leqslant 1 + \sum\limits_{\x\in\mathcal{X}_B^{<n}\setminus\lbrace\0\rbrace}\rho_B(\x,\x)^\alpha,
	\end{equation}
	where the diagonal element with $\x=\0$ was excluded from the sum and estimated trivially.
	
	In view of \eqref{eq:bound_sB} and \eqref{eq:bound_sn}, and inserting the expression for $\rho_B(\x,\x)$, it remains in both cases $s=\vert B\vert$ and $s=n$ to estimate
	\begin{equation}\label{eq:estimate_rho}
		\sum\limits_{\x\in\mathcal{X}_B^{<s}\setminus\lbrace\0\rbrace}\rho_B(\x,\x)^\alpha
			= \sum\limits_{j=1}^{s-1}\sum\limits_{\x\in\mathcal{X}_B^j}
			\left(\sum\limits_{k=1}^{n-j}\sum\limits_{\z\in\mathcal{X}_{B^c}^k}
			\vert\psi(\lbrace\x,\z\rbrace)\vert^2\right)^\alpha.
	\end{equation}
	Let now $\psi_i$ be the orthogonal projection of $\psi$ onto the subspace $\mathcal{H}_\Lambda^i$ with exactly $i$ down spins. For $\x\in\mathcal{X}_B^j$ and $\z\in\mathcal{X}_{B^c}^k$, we can then bound
	\begin{equation}\label{eq:estimate_psi}
		\vert\psi(\lbrace\x,\z\rbrace)\vert^2 \leqslant N_I(\lbrace\x,\z\rbrace)\Vert\psi_{j+k}\Vert^2
			\leqslant C e^{-\mu\, d(\lbrace\x,\z\rbrace,\mathcal{C})} \Vert\psi_{j+k}\Vert^2,
	\end{equation}
	where we used the Cauchy--Schwarz inequality and Lemma~\ref{lemma:density_states}. With the normalization $\sum_i\Vert\psi_i\Vert^2 =1$, the right-hand side of \eqref{eq:estimate_rho} is bounded by
	\begin{equation*}
		\sum\limits_{j=1}^{s-1}\sum\limits_{\x\in\mathcal{X}_B^j}\left(C^\alpha\max\limits_{1\leqslant k\leqslant n-j}
			\sum\limits_{\z\in\mathcal{X}_{B^c}^k}e^{-\alpha\mu\, d(\lbrace\x,\z\rbrace,\mathcal{C})}\right).
	\end{equation*}
	The sum in the parenthesis is over configurations $\lbrace\x,\z\rbrace$ with $j+k\leqslant n$ spins down, of which at least one is in $B$ and one in $B^c$, but without the eventuality of $\x=\B$, corresponding to all spins down on $B$. By the inequality $3 d(\lbrace\x,\z\rbrace,\mathcal{C})\geqslant d(\lbrace\x,\z\rbrace,\mathcal{C}) + d(\x,\mathcal{C}) + d(\x,\partial B)$ followed by a successive application of Lemmas~\ref{lemma:summability_2} and~\ref{lemma:summability_3}, we obtain the bound
	\begin{equation*}\label{eq:final_bound}
		C^\alpha C_1\left(s,\frac{\alpha\mu}{3}\right) \sum\limits_{j=1}^{s-1}\sum\limits_{\x\in\mathcal{X}_B^j}
			e^{-\frac{\alpha\mu}{3}\bigl[d(\x,\mathcal{C}) + d(\x,\partial B)\bigr]}
			\leqslant C^\alpha C_1\left(s,\frac{\alpha\mu}{3}\right)C_2\left(\frac{\alpha\mu}{3}\right)(s-1),
	\end{equation*}
	which concludes the proof since $C_1(s,\mu)$ is affine linear in $s$.
\end{proof}

\subsection{Area law with disorder}

The proof of Theorem~\ref{thm:entropy_random} essentially relies on that of Theorem~\ref{thm:entropy} above, though augmented by an additional decay, on average, of $N_I(\x)$ in the number of down spins of $\x$. In fact, the following general lemma holds.
\begin{lemma}[cf.~Lemma~1.2 in~\cite{BW17}]\label{lemma:decay_n}
	Let $J\subset\mathbbm{R}$ be a compact interval and $\lbrace b_x\rbrace_{x\in\Lambda}$ a collection of iid random variables with non-trivial support $\supp b_x\subset[0,\infty)$. Then, there exist $C,c\in(0,\infty)$ such that for all $n\geqslant 0$, all $\Lambda\subset\mathbbm{Z}$ and all $\x\in\chiln$
	\begin{equation}\label{eq:decay_n}
		\mathbbm{E}\left[N_J(\x)\right]\leqslant C e^{-c n}.
	\end{equation}
\end{lemma}
\begin{proof}[Proof of Theorem~\ref{thm:entropy_random}]
	By the monotonicity of the R\'enyi entropy, we have $S_\alpha \leqslant S_\epsilon$ and the expression to be bounded reduces to
	\begin{equation*}\label{eq:bound_random1}
		\mathbb{E}\left[\max\limits_{\substack{\psi\in\hd,\Vert\psi\Vert =1}}\tr\bigl[(\rho_B)^\epsilon\bigr]\right]
	\end{equation*}
	for some $\epsilon\in(0,1)$. Following the same initial steps as in the proof of Theorem~\ref{thm:entropy}, we conclude for any $\psi\in\hd$ with $\Vert\psi\Vert=1$ that
	\begin{align}
		\tr\bigl[(\rho_B)^\epsilon\bigr]
			&\leqslant\, 3 + \sum\limits_{j=1}^{\vert B\vert-1}\sum\limits_{\x\in\mathcal{X}_B^j}
			\left(\sum\limits_{k=1}^{\vert\Lambda\vert-j}\sum\limits_{\z\in\mathcal{X}_{B^c}^k}
			\vert\psi(\lbrace\x,\z\rbrace)\vert^2\right)^\epsilon\notag\\
			&\leqslant\, 3 + \sum\limits_{j=1}^{\vert B\vert-1}\sum\limits_{\x\in\mathcal{X}_B^j}
			\sum\limits_{k=1}^{\vert\Lambda\vert-j}\sum\limits_{\z\in\mathcal{X}_{B^c}^k}
			N_I(\lbrace\x,\z\rbrace)^\epsilon.\label{eq:bound_random2}
	\end{align}
	Now, by Lemmas~\ref{lemma:density_states} and~\ref{lemma:decay_n}, we have for any $\x\in\mathcal{X}_B^j$ and $\z\in\mathcal{X}_{B^c}^k$
	\begin{equation*}
		\mathbbm{E}\left[N_I(\lbrace\x,\z\rbrace)^\epsilon\right]
			\leqslant C^\epsilon e^{-\frac{\epsilon\mu}{2}d(\lbrace\x,\z\rbrace,\mathcal{C})}e^{-\frac{\epsilon c}{2}(j+k)}.
	\end{equation*}
	Taking the expectation value in~\eqref{eq:bound_random2} and applying Lemma~\ref{lemma:summability_2}, the sum over $\x$ and $\z$ is bounded by $C_1\left(j+k,\frac{\epsilon\mu}{2} \right)$ and thus depends linearly on $j+k$. But with the factor $\exp\bigl\lbrace -\frac{\epsilon c}{2}(j+k)\bigr\rbrace$, the sum over $j$ and $k$ is bounded independently of $B$.
\end{proof}

\appendix
\section{Summability results}

In this appendix we collect three technical lemmas on summability properties of factors of the form $\exp\bigl[-\mu\, d(\x,\y)\bigr]$.
\begin{lemma}[cf.~Lemma~B.2 in~\cite{BW17}]\label{lemma:summability_1}
	For any $\mu>0$ and $n\in\mathbb{N}$, we have
	\begin{equation}\label{inequality_sum}
		\sup\limits_{\x\in\mathcal{C}^{(1)}_\mathbbm{Z}}\sum\limits_{\y\in\mathcal{X}^n_\mathbbm{Z}}
			e^{-\mu\, d(\x,\y)}
			\leqslant \frac{1}{1-e^{-\mu}}\left(\prod\limits_{k=1}^{\infty}\frac{1}{1-e^{-k\mu}}\right)^2
			=:C_{\infty}(\mu).
	\end{equation}
\end{lemma}
\begin{lemma}\label{lemma:summability_2}
	For any $\mu>0$, $n\geqslant 2$ and any interval $B\subset\mathbbm{Z}$, we have
	\begin{equation*}
		\sum\limits_{\substack{\x\in\mathcal{X}_\mathbbm{Z}^{n}\\ \x\cap B\neq \0,\B\\ \x\cap B^c\neq \0}}
		e^{-\mu\, d(\x,\mathcal{C})}
		\leqslant 2\, C_\infty\left(\frac{\mu}{3}\right)
			\bigl(\min\lbrace n,\vert B\vert\rbrace + \coth\left(\frac{\mu}{6}\right)\bigr)
		=: C_1\bigl(\min\lbrace n,\vert B\vert\rbrace,\mu\bigr),
	\end{equation*}
	where $\0$ denotes the empty configuration and $\B$ the configuration where every site in $B$ is occupied.
\end{lemma}
\begin{proof}
	The proof requires some preliminaries. Let $s:= \min\lbrace n,\vert B\vert\rbrace$. For clustered configurations $\w\in\mathcal{C}$ we define the index $I(\w)$ as the length of the shortest part of $\w$ left or right (outside) of $B$. If $s=\vert B\vert$, there are at most two $\w\in\mathcal{C}\cap\chiln$ for each $I(\w)=0,1,2,\ldots,\lfloor\frac{n-\vert B\vert}{2}\rfloor$.
	
	For any $\x\in\chiln$, there exists a (possibly not unique) $\w\in\mathcal{C}$ such that $d(\x,\mathcal{C}) = d(\x,\w)$. Let $\w(\x)$ denote the left-most such configuration, i.e., the one with lowest $w_1$ and let $[\w]$ be the equivalence class of configurations $\x\in\chiln$ such that $\w(\x)=\w$. The configuration space is thus partitioned into a union over clustered configurations $\w\in\mathcal{C}$ followed by a union over elements $\x\in[\w]$.
	
	Finally, observe that since $\x\cap B\neq \0,\B$ and $\x\cap B^c\neq \0$ the following inequality holds:
	\begin{equation*}
		3 d(\x,\mathcal{C}) = 3 d(\x,\w) \geqslant d(\x,\w) + d(\w,\partial B) + I(\w).
	\end{equation*}
	Hence we have
	\begin{equation}\label{eq:sum_d_I}
		\sum\limits_{\substack{\x\in\mathcal{X}_\mathbbm{Z}^{n}\\ \x\cap B\neq \0,\B\\ \x\cap B^c\neq \0}}
		e^{-\mu\, d(\x,\mathcal{C})}
			\leqslant \sum\limits_{\w\in\mathcal{C}}\sum\limits_{\substack{\x\in[\w]\\ \x\textnormal{ as before}}}
				e^{-\frac{\mu}{3}\bigl[d(\x,\w) + d(\w,\partial B) + I(\w)\bigr]}.
	\end{equation}
	By Lemma~\ref{lemma:summability_1}, the sum over $\x\in[\w]$ may be bounded by $C_\infty(\mu/3)$. The sum over $\w\in\mathcal{C}$ can now be split into those $\w$ with $\w\cap B =\B$ and those with $\w\cap B\neq \B$. In the first case, $d(\w,\partial B) = 0$ and in the second case $I(\w) = 0$. Note furthermore that the first case only arises if $s=\vert B\vert$. The right-hand side of~\eqref{eq:sum_d_I} is therefore bounded by $C_\infty(\mu/3)$ times
	\begin{equation*}
			\sum\limits_{\substack{\w\in\mathcal{C}\\ \w\cap B=\B}} e^{-\frac{\mu}{3} I(\w)} \delta_{s,\vert B\vert}
			+ \sum\limits_{\substack{\w\in\mathcal{C}\\ \w\cap B\neq\B}} e^{-\frac{\mu}{3} d(\w,\partial B)}
			\leqslant 2\bigl(s + \coth(\mu/6)\bigr).
	\end{equation*}
	The last inequality may be obtained as follows. For $s=\min\lbrace n,\vert B\vert\rbrace =\vert B\vert$, the first sum is bounded by $2/(1-e^{-\frac{\mu}{3}})$ and the second by $2\vert B\vert +2e^{-\frac{\mu}{3}}/(1-e^{-\frac{\mu}{3}})$. In fact, at most $2\vert B\vert$ clustered configurations $\w$ have $d(\w,\partial B)=0$ and all the others are either left or right of $B$ with at most two of them for each $d(\w,\partial B) = k\in\mathbbm{N}$. For $s=n<\vert B\vert$, the first sum does not contribute and the second is bounded by $2(n+1) + 4e^{-\frac{\mu}{3}}/(1-e^{-\frac{\mu}{3}})$, since here the clustered configurations $\w$ may be left or right of \emph{each boundary} of $B$.
\end{proof}
\begin{lemma}\label{lemma:summability_3}
	For any $\mu>0$, $n\in\mathbbm{N}$ and any interval $\Lambda\subset\mathbbm{Z}$, we have
	\begin{equation}
		\sum\limits_{\x\in\chiln} e^{-\mu \bigl[d(\x,\mathcal{C})+d(\x,\partial\Lambda)\bigr]}
			\leqslant \frac{2 C_\infty\left(\frac{\mu}{2}\right)}{1-e^{-\frac{\mu}{2}}} =: C_2(\mu).
	\end{equation}
\end{lemma}
\begin{proof}
	Using the same decomposition of $\chiln$ into equivalence classes $[\w]$ with $\w\in\mathcal{C}$ as in the proof of Lemma~\ref{lemma:summability_2}, we obtain
	\begin{align*}
		\sum\limits_{\x\in\chiln} e^{-\mu \bigl[d(\x,\mathcal{C})+d(\x,\partial\Lambda)\bigr]}
			&=\, \sum\limits_{\w\in\mathcal{C}}\sum\limits_{\x\in[\w]} e^{-\mu \bigl[d(\x,\w)+d(\x,\partial\Lambda)\bigr]}
			\leqslant \sum\limits_{\w\in\mathcal{C}}\sum\limits_{\x\in[\w]} 
				e^{-\frac{\mu}{2}\bigl[d(\x,\w)+d(\w,\partial\Lambda)\bigr]},
	\end{align*}
	where the last inequality follows by $d(\x,\w) + d(\x,\partial\Lambda)\geqslant d(\w,\partial\Lambda)$. The claim is now an immediate consequence of Lemma~\ref{lemma:summability_1}.
\end{proof}

\bigskip
\noindent\textbf{Acknowledgments.} The research of V.B.~was partially supported by the Swiss National Science Foundation (Grant No.~P2EZP2-162235), which is gratefully acknowledged. 
\bibliography{XXZ_biblio}
\bibliographystyle{abbrv}
\end{document}